%% file: ms.tex
\RequirePackage{amsmath}
\RequirePackage{amssymb}
\RequirePackage{graphicx}
\RequirePackage{tikz}
\documentclass[runningheads]{llncs}

\newenvironment{nonumtheorem*}[1]
    {\renewcommand{\nonumtheoremname}{#1}%
    \begin{nonumtheorem}}
    {\end{nonumtheorem}}

\DeclareMathOperator{\EC}{EC}
\DeclareMathOperator{\D}{D}
\DeclareMathOperator{\poly}{poly}
\DeclareMathOperator{\psens}{psens}
\DeclareMathOperator{\bs}{bs}
\DeclareMathOperator{\C}{C}

\newcommand{\OR}{\mathtt{OR}}
\newcommand{\ADDR}{\mathtt{ADDR}}
\newcommand{\ADDRESS}{\mathtt{ADDRESS}}
\newcommand{\EADDR}{\mathtt{EADDR}}
\newcommand{\AND}{\mathtt{AND}}

\usetikzlibrary{shapes.geometric,shapes.arrows,arrows.meta}
\usetikzlibrary{decorations.pathreplacing,decorations.markings}

\newcommand\ccgate{}
\def\ccgate[#1](#2,#3)(#4){
    \node[scale=1.5] at (#2,#3) {#4};
    \node[circle, draw=black, thick, minimum size=0.54cm] (#1) at (#2,#3) {};
}
\newcommand\sqrgate{}
\def\sqrgate[#1](#2,#3)(#4){
    \node[scale=1.2] (#1) at (#2,#3) {#4};
    \node[rectangle, draw=black, thick, minimum height=0.5cm, minimum width=0.5cm] 
        (#1) at (#2,#3) {};
}
\newcommand\andgate{}
\def\andgate[#1](#2,#3){\ccgate[#1](#2,#3)($\wedge$);}
\newcommand\orgate{}
\def\orgate[#1](#2,#3){\ccgate[#1](#2,#3)($\vee$);}
\newcommand\neggate{}
\def\neggate[#1](#2,#3){\ccgate[#1](#2,#3)($\neg$);}

\begin{document}
\title{On the Relationship between Energy Complexity and other Boolean Function Measures%\thanks{Supported by organization x.}
}
\titlerunning{Energy Complexity and other Boolean Function Measures}
\author{Xiaoming Sun\inst{1,2} \and
Yuan Sun\inst{1,2} \and
Kewen Wu\inst{3} \and
Zhiyu Xia\inst{1,2}}
\authorrunning{Xiaoming Sun et al.}
\institute{CAS Key Lab of Network Data Science and Technology, Institute of Computing Technology, Chinese Academy of Sciences, Beijing, China\\
\and
University of Chinese Academy of Sciences, Beijing, China\\
\email{\{sunxiaoming,sunyuan2016,xiazhiyu\}@ict.ac.cn}\\
\and
School of Electronics Engineering and Computer Science, Peking University, China\\
\email{shlw\_kevin@pku.edu.cn}}
\maketitle              % typeset the header of the contribution
\begin{abstract}
In this work we investigate {\em energy complexity}, a Boolean function measure related to circuit complexity. 
Given a circuit $\mathcal{C}$ over the standard basis $\{\vee_2,\wedge_2,\neg\}$, the energy complexity of $\mathcal{C}$, denoted by $\EC(\mathcal{C})$, is the maximum number of its activated inner gates over all inputs. The energy complexity of a Boolean function $f$, denoted by $\EC(f)$, is the minimum of $\EC(\mathcal{C})$ over all circuits $\mathcal{C}$ computing $f$. 

%This concept has attracted lots of attention in literature.
Recently, K. Dinesh et al. ~\cite{dinesh2018new}
gave $\EC(f)$ an upper bound in terms of the decision tree complexity, $\EC(f)=O(\D(f)^3)$. They also showed that $\EC(f)\leq 3n-1$, where $n$ is the input size. For the lower bound, they show that $\EC(f)\ge\frac{1}{3}\psens(f)$, where $\psens(f)$ is the \textit{positive sensitivity}.  They asked whether $\EC(f)$ can be lower bounded by a polynomial of $\D(f)$.
%Recall that the minimum size of circuit to compute $f$ could be as large as $2^n/n$. 
We improve both the upper and lower bounds in this paper. For upper bounds, We show that $\EC(f)\leq \min\{\frac{1}{2}\D(f)^2+O(\D(f)),n+2\D(f)-2\}$. For the lower bound, we answer Dinesh et al.'s question by proving that $\EC(f)=\Omega(\sqrt{\D(f)})$. For non-degenerated functions, we also give another lower bound $\EC(f)=\Omega(\log{n})$ where $n$ is the input size. These two lower bounds are incomparable to each other. Besides, we examine the energy complexity of $\OR$ functions and $\ADDRESS$ functions, which implies the tightness of our two lower bounds respectively. In addition, the former one answers another open question in~\cite{dinesh2018new}
asking for non-trivial lower bound for energy complexity of $\OR$ functions. 

\keywords{Energy complexity \and Decision tree \and Boolean function \and Circuit complexity.}
\end{abstract}
\section{Introduction}
	
	\subsection{Background and Main Problem}
	
	Given a gate basis $\mathcal B$ and a circuit $\mathcal{C}$ over $\mathcal B$, the {\it energy complexity} of $\mathcal{C}$, defined as $\EC_{\mathcal B}(\mathcal{C})$, is the maximum number of activated gates (except input gates) in $\mathcal{C}$ over all possible inputs. Spontaneously, the {\it energy complexity} of a Boolean function $f:\{0,1\}^n\to\{0,1\}$ over gate basis $\mathcal B$ is defined as $\EC_{\mathcal B}(f):=\min_{\mathcal{C}}\EC_\mathcal{B}(\mathcal{C})$, where $\mathcal{C}$ is a circuit over $\mathcal B$ computing $f$.
	
	When $\mathcal B$ is composed of {\it threshold} gates, this model simulates the neuron activity \cite{uchizawa2006computational,uchizawa2008exponential}, as the transmission of a `spike' in neural network is similar with an activated threshold gate in the circuit. A natural question readily comes up: can Boolean functions be computed with rather few activated gates over threshold gate basis? In order to answer this question, or more precisely, to give lower and upper bounds, plenty of studies were motivated \cite{amano2005complexity,hajnal1993threshold,haastad1991power,razborov1993nomega,uchizawa2006computational}. 
	
	Despite motivated by neurobiology in modern world, tracing back into history, this concept is not brand-new. Let $\EC_{\mathcal B}(n)$ be the maximum energy complexity among all Boolean functions on $n$ variables over basis $\mathcal B$, i.e., the maximum $\EC_{\mathcal B}(f)$ among all possible $f:\{0,1\}^n\to\{0,1\}$. Vaintsvaig \cite{vaincvaig1961power} proved that asymptotically, if $\mathcal B$ is finite, the lower and upper bounds of $\EC_{\mathcal B}(n)$ are $n$ and $2^n/n$ respectively. Then this result was further refined by the outstanding work from Kasim-zade \cite{kasimzade1992measure}, which states that $\EC_{\mathcal B}(n)$ could be $\Theta(2^n/n)$, between $\Omega(2^{n/2})$ and $O(\sqrt n2^{n/2})$, or between $\Omega(n)$ and $O(n^2)$.
	
	When it comes to a specific gate basis, a natural thought is to discuss the energy complexity over the standard Boolean basis $\mathcal{B}=\{\vee_2,\wedge_2,\neg\}$. (From now on we use $\EC(f)$ to represent $\EC_{\mathcal{B}}(f)$ for the standard basis.) Towards this, Kasim-zade \cite{kasimzade1992measure} showed that $\EC(f)=O(n^2)$ for any $n$ variable Boolean function $f$ by constructing an explicit circuit, which was further improved by Lozhkin and Shupletsov \cite{lozhkin2015switching} to $4n$ and then $(3+\epsilon(n))n$ where $\lim_{n\to\infty}\epsilon(n)=0$. Yet these results only connect energy complexity with the number of variables in Boolean function while disregarding other important Boolean function measures.
	
	Recently, Dinesh, Otiv, Sarma \cite{dinesh2018new} discovered a new upper bound which relates energy complexity to decision tree complexity, a well-studied Boolean function complexity measure. In fact, they proved that for any Boolean function $f:\{0,1\}^n\to\{0,1\}$, $\frac{\psens(f)}{3}\le \EC(f)\le \min\big\{O(\D(f)^3),3n-1\big\}$ holds, where the function $\psens(f)$ is defined as the positive sensitivity of $f$, i.e., the maximum of the number of sensitive bits $i\in\{1,2,\ldots,n\}$ with $x_i=1$ over all possible inputs $x$ \cite{dinesh2018new}. However, positive sensitivity may give weak energy complexity lower bounds for some rather fundamental functions. For example, the positive sensitivity of $\OR$ function $f(x_1,...,x_n)=x_1\vee...\vee x_n$ is only $1$. Therefore, Dinesh et al. asked 2 open problems on this issue: 
	\begin{enumerate}
	    \item Does the inequality $\D(f)\le \poly(\EC(f))$ always hold?
	    \item Give a non-trivial lower bound of the energy complexity for $\OR$ function.
	\end{enumerate}

	\subsection{Our Results}

	Throughout this paper, we use completely different method to achieve better bounds from both sides, which explores a polynomial relationship between energy complexity and decision tree complexity and answers two open problems asked by Dinesh et al. Furthermore, we also construct an explicit circuit computing $\OR$ function to match this lower bound.
	
	First, in section 3 we show the upper and lower bounds of energy complexity by decision tree complexity:
	
	\begin{theorem}\label{thm-upperbounds1}
	   For any Boolean function $f:\{0,1\}^n\to\{0,1\}$, 
	   $$\EC(f)\le\frac{1}{2}\D(f)^2+O(\D(f)).$$
	\end{theorem}
	
	\begin{theorem}\label{thm-lowerbounds1}
	   For any Boolean function $f:\{0,1\}^n\to\{0,1\}$, 
	   $$\EC(f)=\Omega(\sqrt{\D(f)}).$$
	\end{theorem}
	
	Note that there are many other Boolean function measures which have polynomial relationship with decision tree complexity. 
	According to the results listed in \cite{hatami2010variations}, we can easily derive the following corollary.
	
	\begin{corollary}
	$\EC(f)$ has polynomial relationship with $\bs(f)$, $\deg(f)$ and $\C(f)$, where $\bs(f)$ is block sensitivity, $\deg(f)$ is degree, and $C(f)$ is certificate complexity.
	\end{corollary}
	
	Second, in section 4 we also show upper and lower bounds of energy complexity with respect to the number of variables.
	
	\begin{theorem}\label{thm-upperbounds2}
	   For any Boolean function $f:\{0,1\}^n\to\{0,1\}$, 
	   $$\EC(f)\le n-2+2\D(f)\leq 3n-2.$$
	\end{theorem}
	
	\begin{theorem}\label{thm-lowerbounds2}
	   For any non-degenerated Boolean function $f:\{0,1\}^n\to\{0,1\}$, 
	   $$\EC(f)=\Omega(\log_2 n).$$
	\end{theorem}
	
	Note that these lower bounds are incomparable with each other, since for any non-degenerated Boolean function $f:\{0,1\}^n\to\{0,1\}$, we have $\Omega(\log_2 n)\le \D(f)\le O(n)$ and this result is essentially tight from both sides.
	
	Finally, in order to show the tightness of lower bounds, we examine the energy complexity of two specific function classes: $\OR$ functions and $\mathtt{EXTENDED}$ $\ADDRESS$ functions (see the definition in Section~\ref{sec:pre}).
	
	\begin{proposition}\label{prop-or}
	   For any positive integer $n$, $\EC(\OR_n) = \Theta(\sqrt{n})$.
	\end{proposition}
	
	\begin{proposition}\label{prop-eaddr}
	   For any positive integer $n$ and an arbitrary Boolean function $g:\{0,1\}^n\to\{0,1\}$, $\EC(\EADDR_{n,g}) = \Theta(n)$.
	\end{proposition}
	
	Note that $\D(\OR_n) = n$ and the number of variables in $\EADDR_{n,g}$ is $ n+2^n$, which shows that the lower bounds in Theorem \ref{thm-lowerbounds1} and \ref{thm-lowerbounds2} are tight. The proofs of these two propositions are in the appendix.
	
    %\subsection{Related works}
    %Besides the standard basis of $\{\vee_2,\wedge_2,\neg\}$, there are other work considering about the energy complexity over bases of threshold gates \cite{uchizawa2006computational,uchizawa2010energy,uchizawa2008exponential} and unate gates \cite{uchizawa2011size}.
    %\par In particular, since there exists similarity between human brain and energy complexity over threshold gate basis, a natural question readily comes up: can Boolean functions be computed with rather few activated gates over threshold gate basis? In order to answer this question, or more precisely, to give lower and upper bounds, plenty of studies were motivated \cite{amano2005complexity,hajnal1993threshold,haastad1991power,razborov1993nomega,uchizawa2006computational}. 
    %\par In addition, there are researches concerning the relationship between energy complexity and the largest fan-in gates \cite{suzuki2013energy}. As for application, it is interesting to consider how to balance the error tolerance and the total energy of the circuit \cite{antoniadis2014energy,barcelo2015almost}.
	
	\section{Preliminaries}\label{sec:pre}
	
	In the following context, we denote $(\underbrace{0,\ldots,0}_{i-1},1,\underbrace{0,\ldots,0}_{n-i})$ as $e_i$, $\{1,2,\ldots,n\}$ as $[n]$, and the cardinality of set $S$ as $|S|$ or $\#S$. 
	
	\subsection{Boolean Function}
    
    A Boolean function $f$ is a function mapping $\{0,1\}^n$ to $\{0,1\}$, where $n$ is a positive integer. We say a Boolean function $f:\{0,1\}^n\to\{0,1\}$ \textit{depends on $m$ variables} if there exists $S\subseteq[n],|S|=m$ and for any $i\in S$, there exists $x\in\{0,1\}^n$ such that $f(x)\neq f(x\oplus e_i)$; and when $f$ depends on all $n$ variables, we say $f$ is \textit{non-degenerated}. 
    
    In section \ref{sec:dtc}, we give the definition of decision tree complexity, a Boolean function measure mainly discussed in this paper. As to other Boolean function measures such as block sensitivity, certificate, and degree, see \cite{buhrman2002complexity} for definition.
    
    Next we define a new Boolean function class called $\tt EXTENDED$ $\ADDRESS$ function, which is an extension of the well-known $\ADDRESS$ function.
    
    \begin{definition}\label{defi:addr} 
        Given integer $n$, the address function $\ADDR_n:\{0,1\}^{n+2^n}\to\{0,1\}$ is defined as
		$$
		    \ADDR_n(x_1,..,x_n,y_0,\dots,y_{2^n-1}) =
		        y_{x_1x_2\dots x_n}.
		$$
	\end{definition}
	
	\begin{definition}\label{defi:3} 
	    Given integer $n$ and an arbitrary Boolean function $g:\{0,1\}^n\to\{0,1\}$, we define the extended address function $\EADDR_{n,g}:\{0,1\}^{n+2^n}\to\{0,1\}$ as 
		$$
		    \EADDR_{n,g}(x_1,..,x_n,y_0,\dots,y_{2^n-1}) = \begin{cases}
		        y_{x_1x_2\dots x_n},        & g(x_1,\dots,x_n)=1\\
		        \bar{y}_{x_1x_2\dots x_n},  & g(x_1,\dots,x_n)=0.
		    \end{cases}
		$$
	\end{definition}
	
	\subsection{Boolean Circuits}
	
	A Boolean circuit $\mathcal{C}$ over a basis $\mathcal{B}$ is a directed acyclic graph which has an output gate, input gates with in-degree $0$ representing variables, and other gates among the circuit basis $\mathcal{B}$. 
	
	For convenience, we have several definitions related with circuit gates here. For two gates $u$, $v$ in a Boolean circuit, we say
	\begin{itemize}
	    \item $u$ is an \textit{inner gate} if and only if $u$ is not an input gate.
	    \item $u$ is \textit{activated} under input $x$ if and only if $u$ outputs $1$ when the input of the circuit is $x$.
	    \item $u$ is \textit{deactivated} under input $x$ if and only if $u$ outputs $0$ when the input of the circuit is $x$.
	    \item $u$ is an \textit{incoming gate} of $v$ if and only if $u$ is an input of $v$.
	    \item $u$ \textit{covers} $v$ if and only if there exists a directed path in circuit from $v$ to $u$.
	\end{itemize}
	
	The circuit basis we mainly discuss is the standard basis $\mathcal{B}=\{\vee_2,\wedge_2,\neg\}$, which means $\lor$-gate with fan-in $2$, $\land$-gate with fan-in $2$ and $\neg$-gate with fan-in $1$. The fan-out of all kinds of gates is unlimited. 
	Particularly, a circuit over standard basis is called {\em monotone} if it does not contain any $\neg$-gate. For convenience, from now on, Boolean circuits and energy complexity are over the standard basis if not specified.
	
	In addition, if the input size of circuit $\mathcal C$ is $n$ and $\mathcal C$ computes $g:\{0,1\}^n\to\{0,1\}$, which depends on $m$ variables, we say $\mathcal C$ \textit{depends on $m$ input gates}.
	
	Next we give energy complexity a mathematical definition.
	
	\begin{definition}
		For a Boolean circuit $\mathcal C$ and an input $x$, the \textit{energy complexity} of $\mathcal{C}$ under $x$ (denoted by $\EC(\mathcal{C},x)$) is defined as the number of activated inner gates in $\mathcal{C}$ under the input $x$. Define the energy complexity of $\mathcal{C}$ as $\EC(\mathcal{C})=\max_{x\in\{0,1\}^n}\{\EC(\mathcal{C},x)\}$ and the energy complexity of a Boolean function $f$ as 
		$$\EC(f)=\min_{\substack{\mathcal{C}\mid \mathcal{C}(x)=f(x)\\\forall x\in\{0,1\}^n}}\EC(\mathcal{C}).$$
    \end{definition}

    \begin{remark}
        \label{remark1}
        Without loss of generality, the first and third structures in Fig. \ref{figu:5} are forbidden in the circuit, since we can replace them by the second or fourth one without increasing the energy complexity. So we can assume that any $\neg$-gate in the circuit has a non-$\neg$ incoming gate and any two $\neg$-gates do not share a same incoming gate.
    \end{remark}

    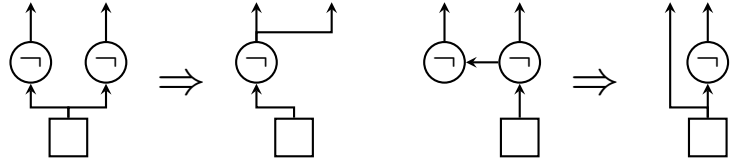
\begin{figure}[ht]
    	\centering
        \input{pics/substructureNot.tex}
    	\caption{Substructures related to $\neg$-gates}
    	\label{figu:5}
    \end{figure}
    
    \subsection{Decision Tree Complexity} \label{sec:dtc}
    
    Following from \cite{buhrman2002complexity}, a (deterministic) decision tree is a rooted ordered binary tree, where each internal node is labeled with a variable $x_i$ and each leaf is labeled with value $0$ or $1$. Given an input $x$, the tree is evaluated as follows. First, we start at the root. Then continue the process until we stop: 
    \begin{itemize}
    \item if we reach a leaf $v$, then stop; 
    \item otherwise query $x_i$, which is the labelled variable on current node, 
        \begin{itemize}
            \item if $x_i=0$, then move to its left child,
            \item if $x_i=1$, then move to its right child. 
        \end{itemize}
    \end{itemize}
    The output of the evaluation is the value on the final position. 
    A decision tree is said to compute $f$, if for any input $x$ the output after evaluation is $f(x)$.
    The complexity of a decision tree $\mathcal T$, denoted by $\D(\mathcal T)$, is its depth, i.e., the number of queries, made on the worst-case input. The decision tree complexity of a Boolean function $f$, denoted by $\D(f)$, is the minimum $\D(\mathcal{T})$ among all decision trees $\mathcal{T}$ computing $f$.

	\section{Upper Bounds of Energy Complexity}

     In this section, we will show upper bounds of energy complexity with respect to decision tree complexity, which improves the result in \cite{dinesh2018new}.

	\begin{nonumtheorem*}{Theorem \ref{thm-upperbounds2} restated}
	   For any Boolean function $f:\{0,1\}^n\to\{0,1\}$, 
	   $$\EC(f)\le n-2+2\D(f)\leq 3n-2.$$
	\end{nonumtheorem*}
	\begin{proof}
	    Since $\D(f)\leq n$, $n+2(\D(f)-1)\leq 3n-2$ holds naturally.
	    Now suppose $\mathcal T$ is a decision tree of $f$ with depth $\D(f)$. Denote the node set of $\mathcal T$ (including leaves) as $S$, where $v_{\it root}\in S$ is the root, $v_{\it left},v_{\it right}\in S$ are the left and right children of $v_{\it root}$ respectively. We also define $F: S\setminus\{v_{root}\} \to S$, where $F(v)$ is the father of node $v$ in $\mathcal T$. Furthermore, define ${\it vbs}: S\to \{x_1,\ldots,x_n\} \cup \{0,1\}$, where ${\it vbs}(v)$ indicates the label on node $v$, i.e., ${\it vbs}(v) = x_i$ means $v$ is labelled with $x_i$ and ${\it vbs}(v) = 0$ (or $1$) means $v$ is a leaf with value $0$ (or $1$). Define $S_0=\{v\in S\mid {\it vbs}(v)=0\}$ and $\widetilde S=S\backslash\left(\{v_{\it root}\}\cup S_0\right)$.\par
	    Based on $\mathcal T$, a circuit $\mathcal{C}$ can be constructed such that $\EC(\mathcal{C})\leq n+2(\D(f)-1)$ as follows. First, define all gates in $\mathcal C$: the input gates are $g_{x_1},\ldots,g_{x_n}$; the $\lnot$-gates are $g^{\lnot}_{x_1},\ldots,g^{\lnot}_{x_n}$ and the $\land$-gates are$g^{\land}_v,v\in\widetilde S$;
	    furthermore, $\mathcal{C}$ contains a unique $\lor$-gate $g^{\lor}$ as output gate with fan-in size $\#\{v\in S\mid {\it vbs}(v)=1\}$. Actually, $g^{\lor}$ is a sub-circuit formed by $\#\{v\in S\mid {\it vbs}(v)=1\}-1$ $\lor$-gates.
	    These gates are connected in following way:
	    \begin{enumerate}
	        \item For all $i\in [n]$, the input of $g^{\lnot}_{x_i}$ is $g_{x_i}$.
	        \item For all $v\in \widetilde S\backslash\{v_{\it left},v_{\it right}\}$, if $v$ is the right child of $F(v)$, the input of $g^{\land}_v$ is $g^{\land}_{F(v)}$ and $g_{{\it vbs}(F(v))}$; otherwise, the input of $g^{\land}_v$ is $g^{\land}_{F(v)}$ and $g^{\lnot}_{{\it vbs}(F(v))}$.
	        \item Merge $g^{\land}_{v_{\it left}}$ with $g^{\lnot}_{\it vbs(v_{\it root})}$; and merge $g^{\land}_{v_{\it right}}$ with $g_{\it vbs(v_{\it root})}$.
	        \item The input of $g^{\lor}$ is all the gates in $\{g^{\land}_v\mid {\it vbs}(v) = 1\}$.
	    \end{enumerate}\par 
	    
	    See Fig. \ref{figu:7} as an example, where
	    $$
	    \begin{aligned}
	    &{\it vbs}(v_{\it root})=x_1\\
	    &{\it vbs}(v_{\it left})=x_2\\
	    &{\it vbs}(v_{\it right})=x_3\\
	    \end{aligned}\qquad
	    \begin{gathered}
	    {\it vbs}(v_1)={\it vbs}(v_2)=x_4\\
	    {\it vbs}(v_3)={\it vbs}(v_5)={\it vbs}(v_8)=1\\
	    {\it vbs}(v_4)={\it vbs}(v_6)={\it vbs}(v_7)=0\\
	    \end{gathered}
	    $$
	    and
	    $S_0=\{v_4,v_6,v_7\},\widetilde S=\{v_{\it left},v_{\it right},v_1,v_2,v_3,v_5,v_8\}$.
	    
		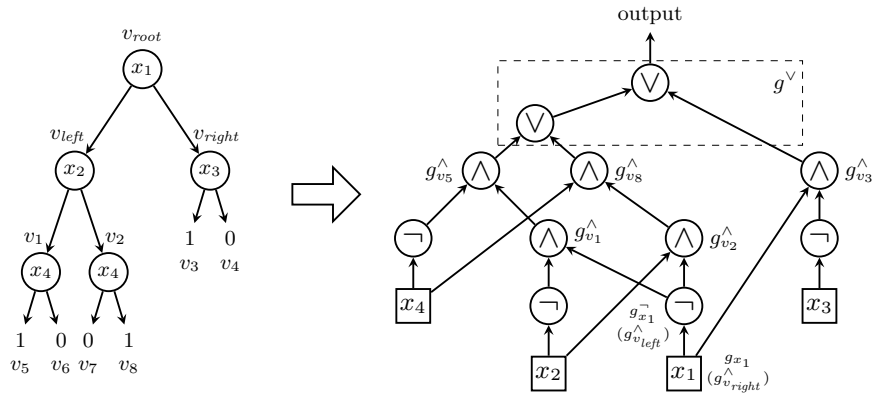
\begin{figure}[ht]
			\centering
            \scalebox{0.9}{\input{pics/rawecD.tex}}
            \caption{Decision tree $\mathcal T$ and circuit $\mathcal{C}$}
            \label{figu:7}
		\end{figure}
		
	    The construction of $\mathcal{C}$ implies several facts:
	    \begin{itemize}
	        \item Under any input, $g^{\land}_{u_1},\ldots,g^{\land}_{u_{k-1}}$ is activated if $g^{\land}_{u_k}$ is activated, where $u_i=F(u_{i+1})$.
	        \item For sibling nodes $u, w\in S$, that $g^{\land}_u$ is activated implies $g^\land_w$ is deactivated, since one of $g^\land_u, g^\land_w$ receives $g_{{\it vbs}(F(u))}$ as input and the other uses $g^\lnot_{{\it vbs}(F(u))}$, which means they can not output $1$ simultaneously.
	    \end{itemize}  
	    These facts imply there are at most $\D(f)-1$ activated $\land$-gate under any input. Furthermore, at most one gate in $\{g^\land_v\mid {\it vbs}(v)=1\}$ is activated. It is easy to construct a circuit computing $\OR_m$ for $g^\lor$, whose energy complexity is no more than $\lceil \log m\rceil$ when promised that the input bits include at most one $1$. Thus, the contribution from $g^\lor$ is no more than
	    $$
	        \left\lceil \log\big(\#\{v\in S\mid {\it vbs}(v) = 1\}\big)\right\rceil \leq \D(f)-1.
	    $$
	    Also, the $\lnot$-gates in $\mathcal{C}$ contribute at most $n$ to the whole energy complexity under any input. Thus, $\EC(\mathcal{C}) \leq n+2(\D(f)-1)$.
	    
	    To justify that circuit $\mathcal{C}$ actually computes $f$, it suffices to show $g^\land_v$ outputs $1$ if and only if $v$ is queried in $\mathcal T$ during the evaluation process under some input. The proof goes as follows:
	    \begin{itemize}
	        \item First, for $v_{\it left}$ and $v_{\it right}$, the claim holds immediately. 
	        \item Then assume that for any node whose depth is less than $k$ in $\mathcal T$, the claim holds. Consider any $v\in\mathcal T$ of depth $k$. Without loss of generality, assume $v$ is the left child of $F(v)$; then the input of $g^\land_v$ is $g^\land_{F(v)}$ and $g^\lnot_{{\it vbs}(F(v))}$. Let $x_i={\it vbs}(F(v))$.
	        \begin{itemize}
	            \item When $g^\land_v$ is activated, $g^\land_{F(v)}$ is activated and $x_i=0$. By induction, $F(v)$ is queried in $\mathcal T$ and the chosen branch after querying is left, which is exactly $v$.
	            \item When $g^\land_v$ is deactivated, either $g^\land_{F(v)}$ is deactivated or $x_i=1$. If it is the former case, then by induction $F(v)$ is not queried; thus $v$ will not as well. Otherwise if $g^\land_{F(v)}$ is activated and $x_i=1$, then $F(v)$ is queried and the chosen branch should be right; thus $v$, which is the left child, will not be queried.
	        \end{itemize}
	    \end{itemize}
	    Thus by induction on the depth of nodes in $\mathcal T$, the claim holds for all $g^\land_v$, which completes the proof of Theorem \ref{thm-upperbounds2}.
	\end{proof}
	
	\begin{nonumtheorem*}{Theorem \ref{thm-upperbounds1} restated}
	   For any Boolean function $f:\{0,1\}^n\to\{0,1\}$, 
	   $$\EC(f)\le\frac{1}{2}\D(f)^2+O(\D(f)).$$
	\end{nonumtheorem*}
	\begin{proof}
	    Suppose $\mathcal T$ is a decision tree of $f$ with depth $\D(f)$. Then by Theorem \ref{thm-upperbounds2}, there is a circuit with energy complexity $n+2(\D(f)-1)$ constructed directly from $\mathcal T$, where $n$ comes from the $\neg$-gates of all variables. 
		In order to reduce the number of $\neg$-gates, we introduce $\D(f)$ additional variables $y_1,y_2,\dots,y_{\D(f)}$ in each level of $\mathcal T$ as a record log of the evaluation process on the tree, where $y_i=0$ means in the $i$-th level of $\mathcal T$, it chooses the left branch, and $y_i=1$ means to choose the right branch. For example, in Fig. \ref{figu:3} these additional variables are computed by $y_1=x_1$, $y_2=\bar{y}_1x_2+y_1x_3$, $y_3=\bar{y}_1\bar{y_2}x_4+\bar{y}_1y_2x_5+y_1\bar{y_2}x_6+y_1y_2x_7$, etc. Given the value of all $y_i$'s, the output of $f$ can be determined by reconstruct the evaluation path in $\mathcal T$; thus $f$ can be viewed as a function on $y_i$'s. Therefore, define 
		$$
		y_{\D(f)+1}=\sum_{z\in\prod\{y_i,\bar y_i\}} f(z)\prod_{i=1}^{\D(f)} z_i.
		$$
		Then given any input $x$, after determine all $y_i$'s ,$y_{\D(f)+1}=f(x)$.
		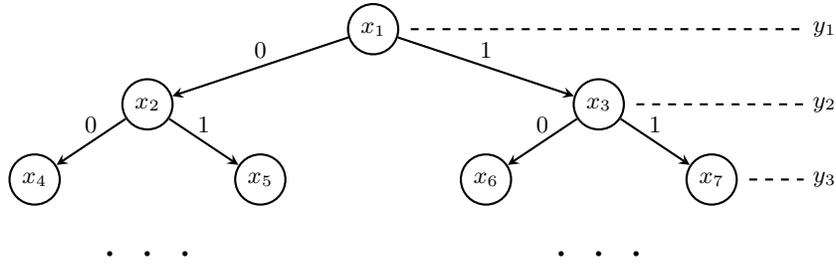
\begin{figure}[ht]
			\centering
                \input{pics/temporaryvar.tex}
			\caption{Decision tree and temporary variables}
			\label{figu:3}
		\end{figure}\par 
		Now construct a circuit using these temporary variables. (See Fig. \ref{figu:4} as an example of the gates for second level of the decision tree.) 
		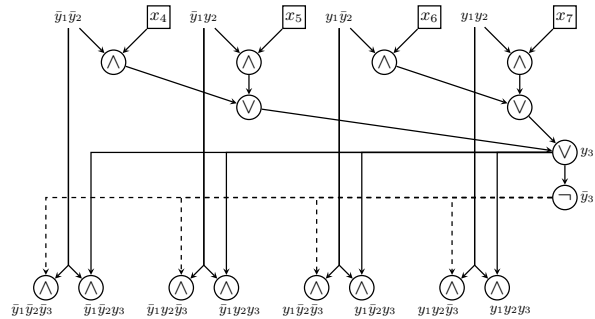
\begin{figure}[ht]
			\centering
            \scalebox{0.6}{\input{pics/subcircuit.tex}}
			\caption{Sub-circuit representing second level of the decision tree}
			\label{figu:4}
		\end{figure}
		Notice that for any $1\leq k\leq\D(f)-1$, to compose $y_{k+1}$, an $\OR_{2^k}$ gadget is required in the $k$-th level sub-circuit, which induces a $k$-level of $\lor$-gates. After computing $y_{k+1}$, we also need two additional levels of gates to compute $\bar y_{k+1}$ and $\prod_{i=1}^{k+1} z_i,z_i\in\{y_i,\bar y_i\}$. In order to compute $y_{\D(f)+1}$, an $\OR_{2^{D(f)}}$ gadget is required, which brings a $\D(f)$-level sub-circuit of $\lor$-gates. Thus summing up all sub-circuits, the circuit depth is $\sum_{i=1}^{\D(f)-1}(i+2)+\D(f)=\frac{1}{2}\D(f)^2+O(\D(f))$.
		
		For any fixed $k,1\le k\le \D(f)$, only one of all $2^k$ cases in $\prod_{i=1}^k z_i,z_i\in\{y_i,\bar y_i\}$ is true, thus each level of the circuit provides at most one activated gate. Then the whole energy complexity is $\frac{1}{2}\D(f)^2+O(\D(f))$.
	\end{proof}
	
	%\cite{hatami2010variations} surveys that $\D(f)=\min\{O(\bs(f)^3),O(\deg(f)^3), O(\C(f)^2)\}$. Thus
	
	%\begin{corollary}
	 % $\EC(f)=\min\{O(\bs(f)^6),O(\deg(f)^6), O(\C(f)^4)\}$.
	%\end{corollary}
	
	\section{Lower Bounds of Energy Complexity}
     
    In this section we will give two theorems on lower bounds of energy complexity. The first one relates decision tree complexity to energy complexity by an intricately constructed decision tree with respect to a given circuit. The second one provides a lower bound depending on the number of variables. In the meantime, we will offer cases where these bounds are tight.
     
	\begin{lemma}\label{lemm:4}
		If $\mathcal{C}$ is a monotone circuit depending on $m$ input gates, then $\EC(\mathcal{C})\ge m-1$.
	\end{lemma}
	\begin{proof}
	    Assume the total number of input gates in $\mathcal C$ is $n$.
	    Since $\mathcal C$ depends on $m$ input gates, the output gate of $\mathcal C$ must cover these gates. Also, the fan-in of $\lor_2,\land_2,\neg$ is at most $2$, thus there are at least $m-1$ inner gates in $\mathcal{C}$. Since $\mathcal{C}$ is monotone, when fed with $1^n$, all inner gates will be activated; therefore $\EC(\mathcal{C})\ge \EC(\mathcal{C},1^n)\ge m-1$.
	\end{proof}
	
	The next lemma is Proposition $2.2$ in \cite{dinesh2018new}. (Also notice Remark \ref{remark1} in Section 2.2.)%For completeness, we give a short proof here.
	
	\begin{lemma}\label{lemm:5}
		If $\mathcal{C}$ is a circuit with $k$ $\neg$-gates, then $\EC(\mathcal{C})\ge k$.
	\end{lemma}
	%\begin{proof}
	%     It suffices to show, when inputted $0^n$, every $\neg$-gate contributes uniquely at least one to energy complexity. Denote $\neg_i$ the $i$-th $\neg$-gate in the following argument.
	%     \begin{itemize}
	%        \item The incoming gate of $\neg_i$ is input gate.
	%            Thus $\neg_i$ is activated immediately.
	%        \item The incoming gate of $\neg_i$ is inner gate.
	%            Thus either $\neg_i$ or its incoming gate is activated.
	%     \end{itemize} 
	%     Since the circuit is free of substructures in left part of Figure \ref{figu:5}, $\EC(\mathcal{C})\ge \EC(\mathcal{C},0^n)\ge k$.
	%\end{proof}
	
	With these preparations, now we can give the main result in this section.
	
	\begin{theorem}\label{theo:7}
		$\EC(f)=\Omega(\sqrt{\D(f)})$ for any Boolean function $f:\{0,1\}^n\to\{0,1\}$, and this lower bound is tight.
	\end{theorem}
	\begin{proof}
		For any Boolean function $f:\{0,1\}^n\to\{0,1\}$ and any circuit $\mathcal{C}$ computing $f$, suppose $\EC(\mathcal{C})=o(\sqrt{\D(f)})$ and let $m$ be the number of $\neg$-gates in $\mathcal{C}$, then $m=o(\sqrt{\D(f)})$ by Lemma \ref{lemm:5}. List all the $\neg$-gates with topological order $\neg_1,\neg_2,\dots,\neg_m$
		such that for any $1\le i<j\le m$, $\neg_i$ does not cover $\neg_j$.
		Suppose the set of all variables (input gates) covered by $\neg_i$ is $\widetilde S_i$, then $S_i$ is defined as $\widetilde S_i\backslash\left(\bigcup_{j=1}^{i-1}\widetilde S_j\right)$. (See also the left side of Fig. \ref{figu:6}.) Define $S_{m+1}=[n]\backslash\left(\bigcup_{j=1}^m\widetilde S_j\right)$. Also let $k_i$ be the number of elements of $S_i$; thus $S_i=\{x_{i,j}\mid j\in[k_i]\}$. Notice that the set collection $S_1, S_2, \dots, S_{m+1}$ is a division of all variables.\par 
        Consider the query algorithm by the order $x_{i,j}$, where $x_{i,j}$ precedes $x_{i',j'}$ if and only if $(i<i')\vee(i=i'\wedge j<j')$. This algorithm induces a decision tree $\mathcal T'$ with depth $n$ immediately. (See also the middle part of Fig. \ref{figu:6}.) 
          	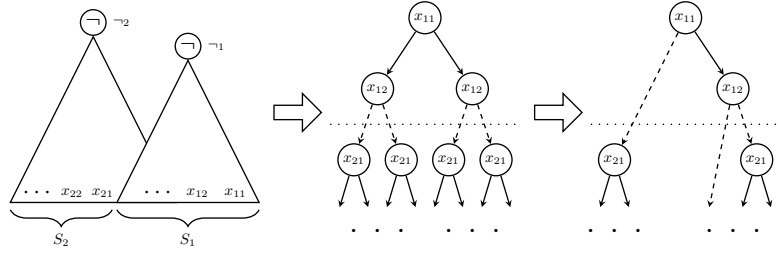
\begin{figure}[ht]
			\centering
            \scalebox{0.9}{\input{pics/inducesimpl.tex}}
			\caption{Circuit $\mathcal{C}$, its induced decision tree $\mathcal T'$, and the simplified decision tree $\mathcal T$}
			\label{figu:6}
		\end{figure}\par
		Since $\mathcal T'$ may be redundant, the simplification process goes as follows: From the root to leaves, check each node whether its left sub-tree and right sub-tree are identical. If so, this node must be inconsequential when queried upon. Thus delete this node and its right sub-tree, and connect its parent to its left child. (See also the right side of Fig. \ref{figu:6}.) \par
		After this process, the new decision tree $\mathcal T$ satisfies: 
        \begin{itemize}
           \item In any path from the root to a leaf, if $x_{i,j}$ is queried before $x_{i',j'}$, $i$ is not greater than $i'$.
           \item Any sub-tree of $\mathcal T$ is non-degenerated, i.e., all queried variables are sensitive in the sub-tree.
           \item The depth of $\mathcal T$ is no smaller than $\D(f)$ since $\mathcal T$ is a decision tree of $f$.
        \end{itemize}\par 
        Let the longest path in $\mathcal T$ be ${\mathcal P}$ and $S_{\mathcal P}$ be the set of variables on ${\mathcal P}$; thus $|S_{\mathcal P}|\ge \D(f)$. Then choose an input $\hat x$ which matches the value of variables on path ${\mathcal P}$.
        
        Suppose $|S_1\cap S_{\mathcal P}|\ge\Omega(\sqrt{\D(f)})$, then the sub-circuit under $\neg_1$ is a monotone circuit depending on at least $|S_1\cap S_{\mathcal P}|$ input gates. Thus the energy complexity in this sub-circuit is $\Omega(\sqrt{\D(f)})$ by Lemma \ref{lemm:4}, which is a contradiction. 
        Therefore $|S_1\cap S_{\mathcal P}|=o(\sqrt{\D(f)})$. Then set variables in $S_1$ to the same value in $\hat x$. Thus the restricted circuit has fewer $\neg$-gates and computes a restricted $f$ function whose decision tree is a sub-tree of $\mathcal T$ with depth at least $|S_{\mathcal P}|-o(\sqrt{\D(f)})$. 
        Now consider $|S_2\cap S_{\mathcal P}|$ in the restricted circuit and the same analysis follows. Continue this restriction process until the value of all $\neg$-gates are determined. 
        
        By then, the depth of the decision tree is still at least $|S_{\mathcal P}|-m\times o(\sqrt{\D(f)})\geq\D(f)-o(\D(f))$. Thus the remaining monotone circuit depends on at least $\D(f)-o(\D(f))$ input gates. By Lemma \ref{lemm:4}, the energy complexity is at least $\D(f)-o(\D(f))=\Omega(\sqrt{\D(f)})$, which is a contradiction.
        
        The tightness of this lower bound is shown in $\EC(\OR_n)=\Theta(\sqrt{n})$ in Proposition \ref{prop-or} as $\D(\OR_n)=n$.
	\end{proof}
	
	%For symmetric function it is easy to get the following corollary.
	
	%\begin{corollary}\label{coro:1}
	%	For any non-constant symmetric Boolean function $f:\{0,1\}^n\to\{0,1\}$, we have $\EC(f)=\Omega(\sqrt{n})$.
	%\end{corollary}

    %\begin{proof}
    %   It is derived immediately from Theorem \ref{theo:7} as $\D(f)=n$ when $f$ is symmetric and non-constant.
    %\end{proof}
    
    Now let us consider the relationship between energy complexity and the input size.
	
	\begin{theorem}\label{theo:6}
		$\EC(f)=\Omega(\log_2{n})$ for any non-degenerated Boolean function $f:\{0,1\}^n\to\{0,1\}$, and this lower bound is tight.
	\end{theorem}
	\begin{proof}
	    Assume $\mathcal{C}$ is an arbitrary circuit computing $f$. It suffices to show $\EC(\mathcal C)\geq\frac12\log_2n$.
	    
		Since $f$ is non-degenerated, the output gate must cover all input gates. Since the fan-in of $\wedge_2,\vee_2,\neg$ gate is no more than $2$, a $k$-depth circuit can cover at most $2^k$ input gates. Thus removing all gates, of which the shortest path to output gate is less than $\log n$ in $\mathcal{C}$, some input gate $x_i$ will be disconnected with the output gate.
		
		Choose an input $\hat x$ satisfying $f(\hat x)=0,f(\hat x\oplus e_i)=1$. Note that when inputted $\hat x$, the output of $\mathcal{C}$ is different from that when inputted $\hat x\oplus e_i$.
		Therefore, there exists a path ${\mathcal P}$ from input gate $x_i$ to output gate under $\hat x$, such that the value of any gate on ${\mathcal P}$ changes after flipping $x_i$. Let $\ell$ be the length of path ${\mathcal P}$, then $\ell\ge\log n$. Also, when inputted $\hat x$ and $\hat x\oplus e_i$, the total number of activated inner gates in ${\mathcal P}$ is $\ell$. It follows immediately
		$$\begin{aligned}
		    \EC(\mathcal{C})
		    &\ge\max\big\{\EC(\mathcal{C},\hat x),\EC(\mathcal{C},\hat x\oplus e_i)\big\}\\
		    &\ge\frac{\big(\EC(\mathcal{C},\hat x)+\EC(\mathcal{C},\hat x\oplus e_i)\big)}{2}
		    \geq\frac{\ell}{2}\ge\frac{\log n}{2}.
		\end{aligned}$$
		
		The tightness of this lower bound is shown by $\EC(\EADDR_{n,g})=\Theta(n)$ in Proposition \ref{prop-eaddr} as $\EADDR_{n,g}:\{0,1\}^{n+2^n}\to\{0,1\}$ is non-degenerated.
	\end{proof}
		
	\section{Conclusion and Open Problems}
    
    Throughout this paper, we build polynomial relationship between energy complexity and other well-known measures of Boolean functions. More precisely, we prove that $\EC(f)\leq \min\{\frac{1}{2}\D(f)^2+O(\D(f)),n+2\D(f)-2\}$ and $\EC(f)=\Omega(\sqrt{\D(f)})$, as well as a logarithmic lower bound in terms of the input size for all non-degenerated function. In addition, we show the tightness of lower bounds by examining $\OR$ functions and $\EADDR$ functions. 
	
    Despite all the effort, some fascinating problems still remain open. 
	\begin{enumerate}
	\item Given two Boolean functions $f$ and $g:\{0,1\}^n\to\{0,1\}$, we say $f$ and $g$ are \textit{isomorphic} if there exists a subset $S\subseteq[n]$ such that $\forall~x\in\{0,1\}^n$,   $f(x)=g(x\bigoplus_{i\in S} e_i)$. Moreover, we denote that $f$ and $g$ are \textit{co-isomorphic} if there exists a subset $S\subseteq[n]$ such that $\forall~x\in\{0,1\}^n$,   $f(x)=1-g(x\bigoplus_{i\in S} e_i)$. Two isomorphic/co-isomorphic Boolean functions have same complexity for most of the Boolean function measures including decision tree complexity, certificate, sensitivity, degree, etc. However, this is not true for energy complexity. For example, $\AND_n$ and $\OR_n$ are co-isomorphic, but there is a quadratic separation between their energy complexity: $\EC(\AND_n)=\Theta(n)$ and $\EC(\OR_n)=\Theta(\sqrt{n})$. What is the largest gap between two co-isomorphic Boolean functions? 
	\item 
	Is the upper bound $\EC(f)=O(\D(f)^2)$ tight?
	We believe this bounds can not be improved. One candidate is the following class of Boolean functions $\mathcal H_n$: 
	a function $h\in\mathcal H_n$ has a non-degenerated depth-$n$ decision tree with $2^n-1$ differently labelled internal nodes, i.e., $h$ is a Boolean function depending on $2^n-1$ variables. Each leaf is assigned to be $0$ or $1$ arbitrarily. Using the similar trick in Proposition \ref{prop-eaddr}, we can show $\forall~ h,g\in\mathcal H_n,\EC(h)\leq 2\EC(g)+2$.  
	The reason why we consider $\mathcal H_n$ is that
	for any Boolean function $f$ with $\D(f)\le n$, there exists an $h\in\mathcal H_n$ such that any circuit $\mathcal{C}$ computing $h$ can be modified into $\mathcal{C}'$, which in turn computes $f$ with $\EC(\mathcal C')\leq\EC(\mathcal{C})$, by rewiring input gates. Therefore, we believe $\forall~h\in\mathcal H_n, \EC(h)=\Theta(n^2)$.
	\end{enumerate}
	
	\section*{Acknowledgement}
	We thank Krishnamoorthy Dinesh for answering some questions with \cite{dinesh2018new}.

\bibliographystyle{splncs04}
\bibliography{reference}

\section*{Appendix}
\appendix
\section{Tight Bounds of Energy Complexity on Specific Functions}
	
	In this section, we discuss two specific function classes, i.e., $\tt EXTENDED$ $\ADDRESS$ function and $\OR_n$ function, on both of which we obtain tight bound of energy complexity. More precisely, $\OR_n$ function shows that the lower bound in Theorem $\ref{theo:7}$ is tight, and $\tt EXTENDED$ $\ADDRESS$ function corresponds with the lower bound in Theorem $\ref{theo:6}$.
	
	\subsection{$\OR_n$ Function}
	
	In this subsection we discuss the energy complexity of the $\OR_n$ function and finish the proof of tightness part in Theorem $\ref{theo:7}$. Given integer $n$, $\OR_n:\{0,1\}^{n}\to\{0,1\}$ is defined as 
	$$
	\OR_n(x_1,x_2,\dots,x_n)=x_1\vee x_2\vee\dots\vee x_n.
	$$

	\begin{nonumtheorem*}{Proposition \ref{prop-or} restated}
	   For any positive integer $n$, $\EC(\OR_n) = \Theta(\sqrt{n})$.
	\end{nonumtheorem*}
	\begin{proof}
		The lower bound follows from $\D(\OR_n)=n$ and Theorem \ref{thm-lowerbounds1}. In order to prove $\EC(\OR_n)=O(\sqrt{n})$, a circuit is constructed as follows (see also Fig. \ref{figu:1}):
		\begin{enumerate}
		    \item Divide all $n$ variables into $\sqrt{n}$ blocks, each block contains $\sqrt{n}$ variables. For variables in the first block, use $\sqrt{n}-1$ $\lor$-gates to connect them as an $\OR_{\sqrt{n}}$ function and mark the output gate of the sub-circuit as $g_1$. 
		    \item Add a $\neg$-gate $h_1$ linked from $g_1$; and for each variable in the second block, feed it into a $\land$-gate together with $h_1$. Then use $\sqrt{n}-1$ $\lor$-gates to connect these $\sqrt n$ $\land$-gates and mark the output gate of the sub-circuit as $g'_2$. 
		    \item Add a $\lor$-gate which has incoming gates $g_1$ and $g'_2$; and insert a $\neg$-gate $h_2$ linked from $g_2$. For each variable in the second block, connect it with $h_2$ by a $\land$-gate. Then use $\sqrt{n}-1$ $\lor$-gates to connect these $\sqrt n$ $\land$-gates. 
		    \item Repeat this process until all blocks are constructed. Then $g_{\sqrt{n}}$ shall be the output gate of the whole circuit.
	    \end{enumerate}
		\begin{figure}[ht]
			\centering
            \scalebox{0.7}{\input{pics/orcircuit.tex}}
			\caption{$\OR_n$ circuit}
			\label{figu:1}
		\end{figure}
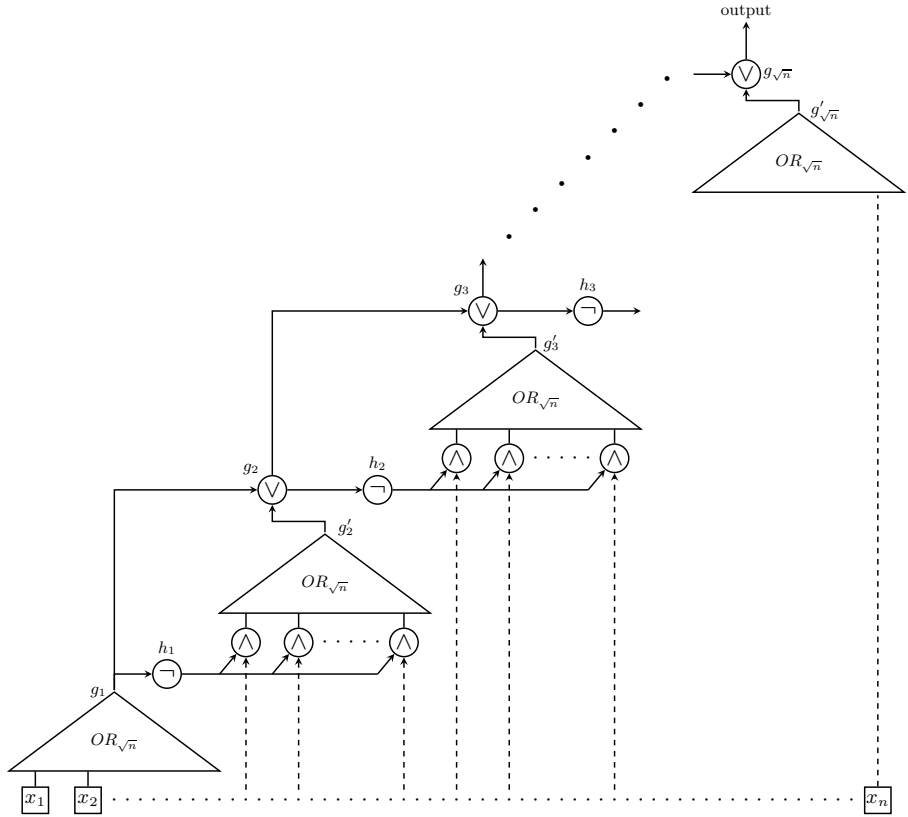\par 
		The main idea of the construction is to view each block as a switch so that if it has an activated gate then it can ``switch of'' all blocks behind it with low cost.
		
		Consider a specific input $x$. If $x=0^n$, then the activated gates are $h_i$'s, whose number is $\sqrt{n}$.
		Otherwise if $x\neq 0^n$, then at least one bit is $1$. Suppose all variables in the first $k-1$ blocks are $0$ and in the $k$-th block there exists a value-$1$ input bit. 
		Then in the first $k-1$ blocks, only $\neg$-gates $h_1, \dots, h_{k-1}$ are activated. 
		And in the $k$-th $\OR_{\sqrt n}$ sub-circuit, at most $\sqrt n-1$ gates are activated. 
		Thus $g_i,i\ge k$ is activated, indicating $h_i,i\ge k$ is deactivated. Therefore, all variables in blocks after $k$-th block are ``switched off''. To sum up, all the activated gates are among $g_i$'s, $g'_i$'s, $h_i$'s, and $k$-th $\OR_{\sqrt n}$ gadget.
		So the energy complexity is $\Theta(\sqrt{n})$.
	\end{proof}
	
	\begin{remark}\label{rema:1}
		\cite{dinesh2018new} showed that $\EC(\AND_n)=\Theta(n)$, so it is rather intriguing that the energy complexity of $\AND_n$ and $\OR_n$ are different while they are basically same under other Boolean function measures such as decision tree complexity, sensitivity, block sensitivity, etc.
	\end{remark}

	\subsection{Extended Address Function}
	
	In this subsection, we discuss the $\tt EXTENDED$ $\ADDRESS$ function, which is defined in Definition \ref{defi:3}, thus complete the proof of Theorem \ref{theo:6}. Note that although $\ADDRESS$ function in itself verifies Theorem \ref{theo:6}, the low-energy circuit for $\ADDR$ actually gives rise to tight bounds of the more generalized $\EADDR_{n,g}$ function.

	\begin{lemma}\label{lemm:addr} For any positive integer $n$, $\EC(\EADDR_{n,g})=\Theta(n)$.
	\end{lemma}
	\begin{proof}
	    The lower bound can be deduced from Theorem \ref{theo:6}. It suffices to prove the upper bound by construction.
	    
	    Let $\mathcal T$ be the natural decision tree of $\ADDR_n$, where all the nodes in the $i$-th ($i\leq n$) level of $\mathcal T$ are labelled with $x_i$, and $y_j$'s are queried in $(n+1)$-th level. Thus, $\mathcal T$ is a full binary tree with depth $n+1$.
	    
	    Now consider the circuit $\mathcal{C}$ constructed in Theorem \ref{thm-upperbounds2} based on $\mathcal T$.
	    Note that the output of $g^\lnot_{y_i}$ are not received by any gate as input. Thus these redundant gates can be safely removed; and the remaining circuit $\mathcal{C}'$ still computes $\ADDR_{n}$.
	    Therefore,
	    $$
	    \EC(\ADDR_n) \leq \EC(\mathcal{C}') \leq 2\big(\D(\ADDR_n)-1\big)+\#\{\neg\text{-gates in }C'\}=3n.
	    $$ 
	\end{proof} 
	
	\begin{nonumtheorem*}{Proposition \ref{prop-eaddr} restated}
	   For any positive integer $n$ and an arbitrary Boolean function $g:\{0,1\}^n\to\{0,1\}$, $\EC(\EADDR_{n,g}) = \Theta(n)$.
	\end{nonumtheorem*}
	\begin{proof}
    	Apply the construction in Lemma \ref{lemm:addr}, and prepare two copies of the circuit computing $\ADDR_n$ and denote them as $\mathcal{C}_0$ and $\mathcal{C}_1$. Then modify them into a new circuit $\mathcal{C}'$ for $\EADDR_{n,g}$ as follows:
		\begin{enumerate}
		    \item For any $x\in\{0,1\}^n$, 
		        \begin{itemize}
		            \item if $g(x)=0$, change the input gate of $y_x$ in $\mathcal{C}_1$ into constant input gate $0$; 
		            \item if $g(x)=1$, change the input gate of $y_x$ in $\mathcal{C}_0$ into constant input gate $1$.
		        \end{itemize}
		    \item For any $i\in[n]$, merge the input gate of $x_i$ in $\mathcal{C}_0,\mathcal{C}_1$ together as the new input gate of $x_i$. 
		    \item Add a $\neg$-gate $\widetilde g$ linked from the output gate of $\mathcal{C}_0$. 
		    \item Add a $\lor$-gate $\widetilde h$ as the new output gate, which takes $\widetilde g$ and the output gate of $\mathcal{C}_1$ as input. 
		\end{enumerate}
		Thus $\mathcal{C}'$ has exactly $n+2^n$ input gates.
		To show $\mathcal{C}'$ actually computes $\EADDR_{n,g}$, it suffices to consider an arbitrary $x\in\{0,1\}^n$.
		If $g(x)=0$, sub-circuit $\mathcal{C}_1$ outputs $y_x$ which becomes $0$ after modification, and $\mathcal{C}_0$ still outputs $y_x$; thus after $\widetilde g$ and $\widetilde h$, $
		C'$ gives $\bar y_x$ correctly. Similar argument holds when $g(x)=1$.
		
		It is also easy to verify that $\EC(\mathcal{C}')$ is bounded by $\EC(\mathcal{C}_0)$ and $\EC(\mathcal{C}_1)$, more precisely,
		$$
		\begin{aligned}
		    \EC(\EADDR_{n,g})&\leq \EC(\mathcal{C}') \leq \max_{x}\big(\EC(\mathcal{C}_0, x) + \EC(\mathcal{C}_1, x)\big) + 2\\
		    &\leq \EC(\mathcal{C}_0)+\EC(\mathcal{C}_1)+2= O(n),
		\end{aligned}
		$$
		which completes the proof combining Theorem \ref{theo:6}.
	\end{proof}
	
\end{document}

%% file: pics/substructureNot.tex
\begin{tikzpicture}[->,thick,>=stealth]
    \sqrgate[in](0,0)();
    \neggate[neg1](-0.5,1);
    \neggate[neg2](0.5,1);
    \draw (in) -- ++(0,0.4) -| (neg1);
    \draw (in) -- ++(0,0.4) -| (neg2);
    \draw (neg1) -- ++(0,0.8);
    \draw (neg2) -- ++(0,0.8);
    
    \node[scale=2] at (1.5,0.7) {$\Rightarrow$};

    \sqrgate[in](3,0)();
    \neggate[neg1](2.5,1);
    \draw (in) -- ++(0,0.4) -| (neg1);
    \draw (neg1) -- ++(0,0.8);
    \draw (neg1) -- ++(0,0.4) -| ++(1,0.4);

    \sqrgate[in](6,0)();
    \neggate[neg1](6,1);
    \neggate[neg2](5,1);
    \draw (in) -- (neg1);
    \draw (neg1) -- (neg2);
    \draw (neg1) -- ++(0,0.8);
    \draw (neg2) -- ++(0,0.8);
    
    \node[scale=2] at (7,0.7) {$\Rightarrow$};

    \sqrgate[in](8.5,0)();
    \neggate[neg1](8.5,1);
    \draw (in) -- (neg1);
    \draw (neg1) -- ++(0,0.8);
    \draw (in) -- ++(0,0.4) -| ++(-0.5,1.4);
\end{tikzpicture}

%% file: pics/rawecD.tex
\begin{tikzpicture}[thick,>=stealth,->]
    \node[circle,draw=black,inner sep=2pt] (x1) at (0,0.5) {$x_1$};
    \node[circle,draw=black,inner sep=2pt] (x2) at (-1,-1) {$x_2$};
    \node[circle,draw=black,inner sep=2pt] (x3) at (1,-1) {$x_3$};
    \node[circle,draw=black,inner sep=2pt] (x4) at (-1.5,-2.5) {$x_4$};
    \node[circle,draw=black,inner sep=2pt] (x44) at (-0.5,-2.5) {$x_4$};
    \node (x31) at (0.7,-2) {$1$}; \node (x30) at (1.3,-2) {$0$}; 
    \node (x41) at (-1.8,-3.5) {$1$}; \node (x40) at (-1.2,-3.5) {$0$};
    \node (x440) at (-0.8,-3.5) {$0$}; \node (x441) at (-0.2,-3.5) {$1$};
    
    \node at (0,1) {$v_{\it root}$}; 
    \node at (-1.1,-0.5) {$v_{\it left}$};
    \node at (1.1,-0.5) {$v_{\it right}$};
    \node at (-1.6,-2) {$v_1$};
    \node at (-0.4,-2) {$v_2$};
    \node at (0.7,-2.4) {$v_3$};
    \node at (1.3,-2.4) {$v_4$};
    \node at (-1.8,-3.9) {$v_5$};
    \node at (-1.2,-3.9) {$v_6$};
    \node at (-0.8,-3.9) {$v_7$};
    \node at (-0.2,-3.9) {$v_8$};

    \draw (x1) -- (x2); \draw (x1) -- (x3); \draw (x2) -- (x4);
    \draw (x2) -- (x44); \draw (x3) -- (x30);
    \draw (x3) -- (x31); 
    \draw (x4) -- (x40); \draw (x4) -- (x41);
    \draw (x44) -- (x440); \draw (x44) -- (x441);

    \node[draw, single arrow, minimum height=10mm, minimum width=8mm, single arrow head extend=2mm, anchor=west]
        at (2.2,-1.35) {};

    \sqrgate[x1](8,-4)($x_1$); \sqrgate[x2](6,-4)($x_2$); \sqrgate[x3](10,-3)($x_3$); \sqrgate[x4](4,-3)($x_4$);
    \neggate[nx1](8,-3); \neggate[nx2](6,-3); \neggate[nx3](10,-2); \neggate[nx4](4,-2);
    \andgate[e1](6,-2); \andgate[e2](8,-2); \andgate[e3](10,-1);
    \andgate[e4](5,-1); \andgate[e5](6.6,-1);
    \orgate[t1](5.8,-0.3); \orgate[t2](7.5,0.3);
    \node (out) at (7.5,1.3) {output};

    \foreach \x in {1,2,3,4}{
        \draw (x\x) -- (nx\x);
    }
    \draw (nx2) -- (e1); \draw (nx1) -- (e1);
    \draw (x2) -- (e2); \draw (nx1) -- (e2);
    \draw (x1) -- (e3); \draw (nx3) -- (e3);
    \draw (nx4) -- (e4); \draw (e1) -- (e4);
    \draw (x4) -- (e5); \draw (e2) -- (e5);
    \draw (e4) -- (t1); \draw (e5) -- (t1);
    \draw (e3) -- (t2); \draw (t1) -- (t2);
    \draw (t2) -- (out);
    
    \node at (6.6,-1.9) {$g^{\land}_{v_1}$}; 
    \node at (8.6,-2) {$g^{\land}_{v_2}$};
    \node at (10.6,-1) {$g^{\land}_{v_3}$}; 
    \node at (4.4,-1) {$g^{\land}_{v_5}$};
    \node at (7.2,-1) {$g^{\land}_{v_8}$};
    \node at (7.4,-3.3) {\footnotesize $\substack{g^{\lnot}_{x_1}\\(g^{\land}_{v_{\it left}})}$};
    \node at (8.8,-4) {\footnotesize $\substack{g_{x_1}\\(g^{\land}_{v_{\it right}})}$};
    
    \node[draw=black,thin,dashed,rectangle,minimum width=4.5cm,minimum height=1.25cm,inner sep=0pt] at (7.5,0) {};
    \node at (9.5,0.3) {$g^{\lor}$};
\end{tikzpicture}

%% file: pics/temporaryvar.tex
\begin{tikzpicture}[->,thick,>=stealth]
    \node[circle,thick,draw=black] (x1) at (0,0) {$x_1$};
    \node (y1) at (6,0) {$y_1$};
    \node[circle,thick,draw=black] (x2) at (-3,-1) {$x_2$};
    \node[circle,thick,draw=black] (x3) at (3,-1) {$x_3$};
    \node (y2) at (6,-1) {$y_2$};
    \node[circle,thick,draw=black] (x4) at (-4.5,-2) {$x_4$};
    \node[circle,thick,draw=black] (x5) at (-1.5,-2) {$x_5$};
    \node[circle,thick,draw=black] (x6) at (1.5,-2) {$x_6$};
    \node[circle,thick,draw=black] (x7) at (4.5,-2) {$x_7$};
    \node (y3) at (6,-2) {$y_3$};

    \draw (x1) -- node[above] {$0$} (x2); \draw (x1) -- node[above] {$1$} (x3);
    \draw (x2) -- node[above] {$0$} (x4); \draw (x2) -- node[above] {$1$} (x5);
    \draw (x3) -- node[above] {$0$} (x6); \draw (x3) -- node[above] {$1$} (x7);
    \draw[-,dashed] (x1) ++ (0.5,0) -- (y1);
    \draw[-,dashed] (x3) ++ (0.5,0) -- (y2);
    \draw[-,dashed] (x7) ++ (0.5,0) -- (y3);

    \foreach \x in {-3.5,-3.0,-2.5,3.5,3.0,2.5}{
        \node[scale=2] at (\x,-3) {$\cdot$};
    }
\end{tikzpicture}

%% file: pics/subcircuit.tex
\begin{tikzpicture}[->,thick,>=stealth]
    \node (y0) at (-1,0) {$\bar y_1\bar y_2$};
    \sqrgate[x0](1,0)($x_4$);
    \node (y1) at (2,0) {$\bar y_1 y_2$};
    \sqrgate[x1](4,0)($x_5$);
    \node (y2) at (5,0) {$y_1\bar y_2$};
    \sqrgate[x2](7,0)($x_6$);
    \node (y3) at (8,0) {$y_1 y_2$};
    \sqrgate[x3](10,0)($x_7$);

    \foreach \x in {0,1,2,3}{
        \andgate[a\x](\x*3,-1);
        \draw (y\x) -- (a\x);
        \draw (x\x) -- (a\x);
    }
    \orgate[ad0](3,-2); \orgate[ad1](9,-2); 
    \orgate[add](10,-3);
    \draw (a0) -- (ad0); \draw (a1) -- (ad0);
    \draw (a2) -- (ad1); \draw (a3) -- (ad1);
    \draw (ad0) -- (add); \draw (ad1) -- (add);
    \neggate[neg](10,-4); \draw (add) -- (neg); 
    \node at (10.5,-3) {$y_3$};
    \node at (10.5,-4) {$\bar y_3$};

    \foreach \z [evaluate=\z as \x using 3*\z-0.5] in {0,1,2,3}{
        \andgate[u\z](\x,-6); \andgate[v\z](\x-1,-6); 
        \draw (add) -| (u\z); \draw[dashed] (neg) -| (v\z);
    }
    \foreach \x in {0,1,2,3}{
        \draw (y\x) -- ++(0,-5.5) -- (u\x);
        \draw (y\x) -- ++(0,-5.5) -- (v\x);
    }
    \footnotesize
    \node at (-1.8,-6.5) {$\bar y_1\bar y_2\bar y_3$};
    \node at (-0.2,-6.5) {$\bar y_1\bar y_2 y_3$};
    \node at (1.2,-6.5) {$\bar y_1 y_2\bar y_3$};
    \node at (2.8,-6.5) {$\bar y_1 y_2 y_3$};
    \node at (4.2,-6.5) {$y_1\bar y_2\bar y_3$};
    \node at (5.8,-6.5) {$y_1\bar y_2 y_3$};
    \node at (7.2,-6.5) {$y_1 y_2\bar y_3$};
    \node at (8.8,-6.5) {$y_1 y_2 y_3$};
\end{tikzpicture}

%% file: pics/inducesimpl.tex
\scalebox{0.7}{
\begin{tikzpicture}[thick,>=stealth,->,
        tri/.style={draw=black, shape=isosceles triangle, minimum height=3cm, minimum width=3cm,
        shape border rotate=+90,
        isosceles triangle stretches}]

    \node[tri, minimum height=3.5cm, minimum width=3.5cm] at (-2,0.155) {};
    \node[tri, fill=white] at (0,0) {};
    \node at (1,-0.7) {$x_{11}$};
    \node at (0.2,-0.7) {$x_{12}$};
    \node at (-0.6,-0.7) {\Large $\cdots$};
    \node at (-1.8,-0.7) {$x_{21}$};
    \node at (-2.45,-0.7) {$x_{22}$};
    \node at (-3.15,-0.7) {\Large $\cdots$};

    \draw [decorate,decoration={brace,amplitude=10pt,mirror,raise=12pt},yshift=-18pt,-] 
        (-1.5,0) -- (1.5,0) node [black,midway,yshift=-30pt] {$S_1$};
    \draw [decorate,decoration={brace,amplitude=10pt,mirror,raise=12pt},yshift=-18pt,-] 
        (-3.75,0) -- (-1.6,0) node [black,midway,yshift=-30pt] {$S_2$};

    \neggate[neg1](0,2.4); 
    \node at (0.6,2.4) {$\neg_1$};
    \neggate[neg2](-2,2.9);
    \node at (-1.4,2.9) {$\neg_2$};

    \node[draw, single arrow, minimum height=10mm, minimum width=8mm, single arrow head extend=2mm, anchor=west] at (1.8,1) {};

    \node[circle,draw=black,inner sep=2pt] (n1) at (5,3) {$x_{11}$};
    \node[circle,draw=black,inner sep=2pt] (n21) at (4,1.5) {$x_{12}$};
    \node[circle,draw=black,inner sep=2pt] (n22) at (6,1.5) {$x_{12}$};
    \node[circle,draw=black,inner sep=2pt] (n31) at (3.5,0) {$x_{21}$};
    \node[circle,draw=black,inner sep=2pt] (n32) at (4.5,0) {$x_{21}$};
    \node[circle,draw=black,inner sep=2pt] (n33) at (5.5,0) {$x_{21}$};
    \node[circle,draw=black,inner sep=2pt] (n34) at (6.5,0) {$x_{21}$};
    
    \draw (n1) to (n21);
    \draw (n1) to (n22);
    \draw[dashed] (n21) to (n31);
    \draw[dashed] (n21) to (n32);
    \draw[dashed] (n22) to (n33);
    \draw[dashed] (n22) to (n34);
    \draw[-,loosely dotted] (3,0.75) to (7,0.75);
    \foreach \x in {1,2,3,4}{
        \draw (n3\x) to ++(-0.3,-1);
        \draw (n3\x) to ++(0.3,-1);
    }

    \foreach \x in {3.5,4,4.5,5.5,6,6.5}{
        \node[scale=2]  at (\x,-1.5) {$\cdot$};
    }

    \node[draw, single arrow, minimum height=10mm, minimum width=8mm, single arrow head extend=2mm, anchor=west] at (7.3,1) {};

    \node[circle,draw=black,inner sep=2pt] (n1) at (10.5,3) {$x_{11}$};
    \node[circle,draw=black,inner sep=2pt] (n31) at (9,0) {$x_{21}$};
    \node[circle,draw=black,inner sep=2pt] (n22) at (11.5,1.5) {$x_{12}$};
    \node[circle,draw=black,inner sep=2pt] (n34) at (12,0) {$x_{21}$};

    \draw[dashed] (n1) to (n31);
    \draw (n1) to (n22);
    \draw[dashed] (n22) to (n34);
    \draw[dashed] (n22) to ++(-0.5,-2.5);
    \draw[-,loosely dotted] (8.5,0.75) to (12.5,0.75);
    \foreach \x in {1,4}{
        \draw (n3\x) to ++(-0.3,-1);
        \draw (n3\x) to ++(0.3,-1);
    }

    \foreach \x in {9,8.5,9.5,11.5,11,12}{
        \node[scale=2] at (\x,-1.5) {$\cdot$};
    }

\end{tikzpicture}}

%% file: pics/orcircuit.tex
\begin{tikzpicture}[->,thick,>=stealth,
        smallOr/.style={draw=black, shape=isosceles triangle, minimum height=1.5cm, minimum width=4cm,
                        shape border rotate=+90,
                        isosceles triangle stretches}]
    \node at (0,-0.1) {$OR_{\sqrt n}$};
    \node[smallOr] (or1) at (0,0) {};
    \node at (-0.3,0.85) {$g_1$};
    
    \sqrgate[x1](-1.5,-1.2)($x_1$); \draw[-] (x1) -- ++(0,0.55);
    \sqrgate[x2](-0.5,-1.2)($x_2$); \draw[-] (x2) -- ++(0,0.55);
    \foreach \x in {0.00,0.25,...,14.00}{
        \node at (\x,-1.2) {\large $\cdot$};
    }
    \foreach \x in {4.00,4.25,...,5.00}{
        \node at (\x,1.8) {\large $\cdot$};
    }

    \node at (4,2.9) {$OR_{\sqrt n}$};
    \node[smallOr] (or2) at (4,3) {}; \node at (4.4,4) {$g_2'$};
    \andgate[and1](2.5,1.8); \draw[dashed] (2.5,-1) -- (and1); \draw[-] (and1) -- ++(0,0.55);
    \andgate[and2](3.5,1.8); \draw[dashed] (3.5,-1) -- (and2); \draw[-] (and2) -- ++(0,0.55);
    \andgate[andn](5.5,1.8); \draw[dashed] (5.5,-1) -- (andn); \draw[-] (andn) -- ++(0,0.55);
    \neggate[neg1](1,1.2); \node at (1,1.7) {$h_1$};
    \orgate[sor](3,4.7); \node at (2.6,5.1) {$g_2$};
    \neggate[neg2](5,4.7); \node at (5,5.2) {$h_2$};
    \draw (or1) |- (neg1); 
    \draw (or1) |- (sor);
    \draw[-] (neg1) -- (5,1.2);
    \draw[<-] (and1) -- ++(-0.5,-0.6);
    \draw[<-] (and2) -- ++(-0.5,-0.6);
    \draw[<-] (andn) -- ++(-0.5,-0.6);
    \draw (or2) |- ++(0,1.1) -| (sor);
    \draw (sor) -- (neg2);

    \foreach \x in {4.00,4.25,...,5.00}{
        \node at (\x+4,5.3) {\large $\cdot$};
    }

    \node at (8,6.4) {$OR_{\sqrt n}$};
    \node[smallOr] (or3) at (8,6.5) {};
    \andgate[and1](6.5,5.3); \draw[dashed] (6.5,-1) -- (and1); \draw[-] (and1) -- ++(0,0.55);
    \andgate[and2](7.5,5.3); \draw[dashed] (7.5,-1) -- (and2); \draw[-] (and2) -- ++(0,0.55);
    \andgate[andn](9.5,5.3); \draw[dashed] (9.5,-1) -- (andn); \draw[-] (andn) -- ++(0,0.55);
    \draw[-] (neg2) -- (9,4.7);
    \orgate[sor2](7,8.1); \node at (6.6,8.5) {$g_3$}; \node at (8.3,7.5) {$g_3'$};
    \neggate[neg2](9,8.1); \node at (9,8.6) {$h_3$};
    \draw[<-] (and1) -- ++(-0.5,-0.6);
    \draw[<-] (and2) -- ++(-0.5,-0.6);
    \draw[<-] (andn) -- ++(-0.5,-0.6);
    \draw (or3) |- ++(0,1.1) -| (sor2);
    \draw (sor2) -- (neg2);
    \draw (sor) |- (sor2);
    \draw (sor2) -- ++(0,1);
    \draw (neg2) -- ++(1,0);

    \foreach \x in {8.0,8.5,...,11.0}{
        \node at (\x-0.5,\x+1.5) {\Huge $\cdot$};
    }

    \node at (13,10.9) {$OR_{\sqrt n}$};
    \node[smallOr] (orn) at (13,11) {};
    \orgate[last](12,12.6); \node at (12.6,12.6) {$g_{\sqrt n}$};
    \draw (orn) |- ++(0,1.1) -| (last); \node at (13.5,11.9) {$g_{\sqrt n}'$};
    \draw[<-] (last) -- ++(-1,0); 
    \draw (last) -- node[above,yshift=3mm] {output} ++(0,1);

    \sqrgate[xn](14.5,-1.2)($x_n$); 
    \draw[-,dashed] (xn) -- ++(0,11.5);
\end{tikzpicture}

%% file: ms.bbl
\begin{thebibliography}{10}
\providecommand{\url}[1]{\texttt{#1}}
\providecommand{\urlprefix}{URL }
\providecommand{\doi}[1]{https://doi.org/#1}

\bibitem{amano2005complexity}
Amano, K., Maruoka, A.: On the complexity of depth-2 circuits with threshold
  gates. In: International Symposium on Mathematical Foundations of Computer
  Science. pp. 107--118. Springer (2005)

\bibitem{buhrman2002complexity}
Buhrman, H., De~Wolf, R.: Complexity measures and decision tree complexity: a
  survey. Theoretical Computer Science  \textbf{288}(1),  21--43 (2002)

\bibitem{dinesh2018new}
Dinesh, K., Otiv, S., Sarma, J.: New bounds for energy complexity of boolean
  functions. In: International Computing and Combinatorics Conference. pp.
  738--750. Springer (2018)

\bibitem{hajnal1993threshold}
Hajnal, A., Maass, W., Pudl{\'a}k, P., Szegedy, M., Tur{\'a}n, G.: Threshold
  circuits of bounded depth. Journal of Computer and System Sciences
  \textbf{46}(2),  129--154 (1993)

\bibitem{haastad1991power}
H{\aa}stad, J., Goldmann, M.: On the power of small-depth threshold circuits.
  Computational Complexity  \textbf{1}(2),  113--129 (1991)

\bibitem{hatami2010variations}
Hatami, P., Kulkarni, R., Pankratov, D.: Variations on the sensitivity
  conjecture. arXiv preprint arXiv:1011.0354  (2010)

\bibitem{kasimzade1992measure}
Kasim-Zade, O.M.: On a measure of active circuits of functional elements.
  Mathematical problems of cybernetics  \textbf{4},  218--228 (1992)

\bibitem{lozhkin2015switching}
Lozhkin, S., Shupletsov, M.: Switching activity of boolean circuits and
  synthesis of boolean circuits with asymptotically optimal complexity and
  linear switching activity. Lobachevskii Journal of Mathematics
  \textbf{36}(4),  450--460 (2015)

\bibitem{vaincvaig1961power}
Modest~Nikolaevich, V.: On the power of networks of functional elements. In:
  Proceedings of the USSR Academy of Sciences. vol.~139, pp. 320--323. Russian
  Academy of Sciences (1961)

\bibitem{razborov1993nomega}
Razborov, A., Wigderson, A.: $n^{\Omega(\log n)}$ lower bounds on the size of
  depth-$3$ threshold cicuits with {AND} gates at the bottom. Information
  Processing Letters  \textbf{45}(6),  303--307 (1993)

\bibitem{uchizawa2006computational}
Uchizawa, K., Douglas, R., Maass, W.: On the computational power of threshold
  circuits with sparse activity. Neural Computation  \textbf{18}(12),
  2994--3008 (2006)

\bibitem{uchizawa2008exponential}
Uchizawa, K., Takimoto, E.: Exponential lower bounds on the size of
  constant-depth threshold circuits with small energy complexity. Theoretical
  Computer Science  \textbf{407}(1-3),  474--487 (2008)

\end{thebibliography}
